\documentclass[letterpaper,10pt,conference]{ieeeconf}
\IEEEoverridecommandlockouts
\usepackage{amsfonts,amsmath,amssymb,mathtools,bbm,bm}
\usepackage{diagbox}
\usepackage{hyperref}
\usepackage{graphicx,caption,subcaption}
\usepackage{booktabs}
\usepackage{algorithm}
\usepackage[noend]{algpseudocode}
\newdimen{\algindent}
\usepackage{cite}
\usepackage{xcolor}
\newtheorem{theorem}{Theorem}

\newcommand{\edit}[1]{#1}
\algnewcommand{\LeftComment}[2]{\Statex \hspace{#1\algindent} \textcolor{olive}{\# #2}}
\algnewcommand{\Comm}[1]{\hspace{0.5\algindent} \textcolor{olive}{\# #1}}

\def\cD{{\cal D}}

\def\cK{{\cal K}}

\def\cR{{\cal R}}

\def\cZ{{\cal Z}}

\newcommand{\eg}[0]{\textit{e.g. }}
\newcommand{\ie}[0]{\textit{i.e. }}

\title{\LARGE \bf
Multi-agent Path Finding for Cooperative Autonomous Driving
}

\author{Zhongxia Yan, Han Zheng, Cathy Wu 
\thanks{This work was supported by the National Science Foundation (NSF) CAREER award 2239566 and the MIT Amazon Science Hub.}
\thanks{Zhongxia Yan, Han Zheng and Cathy Wu are with the Laboratory for Information \& Decision Systems (LIDS), Massachusetts Institute of Technology, Cambridge, MA 02139, USA. Email:\{{\tt\small zxyan,hanzheng,cathywu\}@mit.edu}}
}

\begin{document}

\maketitle
\thispagestyle{empty}
\pagestyle{empty}

\begin{abstract}
Anticipating possible future deployment of connected and automated vehicles (CAVs), cooperative autonomous driving at intersections has been studied by many works in control theory and intelligent transportation across decades. Simultaneously, recent parallel works in robotics have devised efficient algorithms for multi-agent path finding (MAPF), though often in environments with simplified kinematics. In this work, we hybridize insights and algorithms from MAPF with the structure and heuristics of optimizing the crossing order of CAVs at signal-free intersections. We devise an optimal and complete algorithm, Order-based Search with Kinematics Arrival Time Scheduling (OBS-KATS), which significantly outperforms existing algorithms, fixed heuristics, and prioritized planning with KATS. The performance is maintained under different vehicle arrival rates, lane lengths, crossing speeds, and control horizon. Through ablations and dissections, we offer insight on the contributing factors to OBS-KATS's performance. Our work is directly applicable to many similarly scaled traffic and multi-robot scenarios with directed lanes.
\end{abstract}

\section{Introduction}
\label{sec:introduction}
The development of autonomous driving technology raises the possibility of intelligent coordination of connected and automated vehicles (CAVs) towards societal objectives, such as reducing congestion and fuel consumption, as well as improving safety. Therefore, many works on intelligent transportation systems \cite{dresner2008multiagent,wu2021flow,yan2022unified,xu2021comparison} have studied potential positive impacts of cooperative driving of CAVs. In particular, signal-free intersections are regions where coordination of CAVs is critical to safety and efficiency. These intersections are not restricted to intelligent transportation systems, but also are commonly found in real-world robotic warehouses at crossings between directed lanes \cite{li2021lifelong}. In this work, we adapt insights and algorithms from multi-agent path finding (MAPF) for the coordination of a cooperative driving intersection with rich vehicle kinematics. Like \cite{xu2021comparison}, we divide the overall task of coordinating CAVs into two sequential phases, first optimizing the CAV crossing order then computing order-conditioned vehicle trajectories. To define the crossing order, we divide the intersection into a reservation system where the arrival and departure times at \textit{subzones} are planned by our low-level Kinematic Arrival Time Scheduling (KATS), which substitutes for a low-level path planning algorithm. While existing works optimize the crossing order with First-In-First-Out (FIFO) heuristics \cite{dresner2008multiagent} and Monte Carlo Tree Search (MCTS) \cite{xu2019cooperative}, we demonstrate that our MAPF-inspired high-level prioritized planning \cite{erdmann1987multiple} and Order-based Search (OBS) algorithms obtain significantly superior solution quality. We obtain order-conditioned vehicles trajectories with trajectory optimization rather than single agent path planning algorithms like A* search \cite{hart1968formal} or SIPP \cite{phillips2011sipp}, allowing us to bypass kinematics limitations.

In summary, our main contributions are:\begin{enumerate}
    \item Incorporating insights from MAPF, we design an algorithm for ordering vehicle crossings at a signal-free intersection and translating the crossing order to vehicle trajectories within a kinematic bicycle model.
    \item We empirically characterize the OBS-KATS's significant improvement of vehicle delays over baselines under a wide range of intersection settings.
    \item We prove the soundness, completeness, and optimality of OBS-KATS for finding vehicle crossing orders.
\end{enumerate}
\edit{We provide full source code for reproducibility on \href{https://github.com/mit-wu-lab/mapf-autonomous-driving}{GitHub}}.

\section{Related Work}
\label{sec:related_work}
\subsection{Cooperative Driving at Intersections}
Cooperative driving of connected and automated vehicles (CAVs) has been studied in intelligent transportation settings ranging from adaptive cruise control \cite{van2006impact} to traffic networks with diverse structures \cite{yan2022unified}. In particular, several recent cooperative driving strategies have been proposed for optimizing the crossing order of CAVs at signal-free intersections \cite{xu2021comparison}. The First-In-First-Out (FIFO) strategy has been studied by \cite{dresner2008multiagent} as a heuristic crossing order. Given an existing crossing order, the Dynamic Resequencing method \cite{zhang2018decentralized} inserts a newly arriving vehicles into a suitable position, but keeps the rest of the order unchanged. On the other hand, \cite{xu2019cooperative} demonstrates that Monte Carlo Tree Search (MCTS) can be used to obtain a more optimal crossing order by periodically replanning the existing order. Our work significantly improves upon these previous methods in the cooperative driving setting by leveraging insights and algorithms from multi-agent path finding.

\subsection{Multi-agent Path Finding}
The classical multi-agent path finding (MAPF) problem \cite{stern2019multi} is a NP-hard \cite{yu2013structure} problem which seeks to find the shortest collision-avoiding paths for a set of agents in a discrete graph. Since the space of joint agent trajectories is intractably large to consider \cite{sharon2015conflict}, nearly all MAPF algorithms rely on repeatedly calling single-agent path planner such as A* search \cite{hart1968formal} or SIPP \cite{phillips2011sipp}, while holding paths of some set of other agents as constraints.

Prioritized planning (PP) \cite{erdmann1987multiple,silver2005cooperative} plans one agent trajectory at a time in random agent order while avoiding collisions with all previously planned trajectories. Conflict-based Search (CBS) \cite{sharon2015conflict} is a seminal solver which relies on backtracking tree-search to resolve pairwise agent collisions, and Priority-based Search (PBS) \cite{ma2019searching} is a scalable extension of CBS, albeit suboptimal and incomplete. We derive significant algorithmic insights from these works.

Recent methods have aimed at improving the solution quality \cite{li2021anytime,huang2022anytime} and completeness \cite{okumura2022priority,li2022mapf,okumura2023improving} under large-scale settings with up to thousands of agents.

\subsection{Continuous MAPF and Multi-robot Motion Planning}
As classical MAPF is discrete time and space, continuous settings may be discretized for application of MAPF algorithms \cite{honig2018trajectory}. Recent continuous MAPF works have investigated planning with continuous time directly \cite{andreychuk2022multi,andreychuk2021improving,kasaura2022prioritized}, gut require simplified agent kinematics such as constant speed along graph edges. \edit{Relatedly, \cite{li2023intersection} applies MAPF to intersection traffic settings with unbounded acceleration}. Finally, works in multi-robot motion planning \cite{kottinger2022conflict,okumura2022quick} have applied sampling based methods like probabilistic roadmaps \cite{kavraki1996probabilistic} to plan over settings with continuous 2D space and time. As traffic systems typically contain well-defined lanes, formulating our problem with continuous 2D space is unnecessary.

\section{Problem Formulation}
\label{sec:formulation}
We formulate the cooperative driving problem at a single intersection, though this formulation is applicable to any single-junction traffic scenario (e.g. highway merging \cite{yan2022unified}). Consider a four-way intersection with directions $i \in \cD = \{1, 2, 3, 4\}$; along each direction, a single \textit{entering} lane (towards intersection) and \textit{exiting} lane has length $\ell_\text{lane}$. Vehicle routes $r = (i, j) \in \cR = \cD^2$ are considered, and a vehicle $k$ with length $\ell_k$ traveling along route $r$ passes the intersection from direction $i$ to direction $j$, either heading straight or making a left- or right-turn. If space is available, vehicle $k$ enters the system from an entering lane at a \edit{deterministic} rate $\lambda_i$ (veh/hr/lane) with initial speed $v_0$, its route is sampled according to $r_k \sim P(r = (i, \cdot))$ to account for different turn probabilities. Towards collision avoidance, we design a division of the intersection into 16 reservation subzones $z \in \cZ$ (Fig.~\ref{fig:intersection_subzones}), which may only be occupied by one vehicle at a given time, based on geometries of crossing vehicle routes. While the subzone design in \cite{xu2019cooperative,xu2021comparison} does not permit simultaneous left turns, our design permits four turning vehicles (two left-turn and two right-turn) vehicles to pass the intersection simultaneously. Longitudinal position along route $r$ is defined in the range $[0, \ell_r]$. The position, speed, and acceleration of a vehicle $k$ at step $t$ is denoted as $x_k(t)$, $v_k(t)$, and $a_k(t)$, respectively. The start and end positions of each subzone $z$ along each passing route $r$ are denoted as $x_{z, r, 0}$ and $x_{z, r, 1}$, respectively. Vehicles are subjected to maximum straight speed $\overline v$ and maximum turning speed $\overline v_{r, z} \leq \overline v$ in a subzone, as well as acceleration limits $[\underline a, \overline a]$. We assume perfect sensing, inter-vehicle communication, and control. 

Like works before us \cite{xu2021comparison}, the objective at each planning step is to find the ordering of vehicles crossing the intersection which minimizes total vehicle delay, which is defined as the difference between travel time $\sum_k t(x_k \geq \ell_{r_k})$ and minimum travel time $\sum_k \underline t(x_k \geq \ell_{r_k})$ absent of other vehicles; we use the notation $t(x_k \geq x)$ to denote the first time such that $x_k \geq x$. The crossing order is a partial ordering which defines precedence relationships for vehicles whose routes cross the same reservation subzone, but not vehicles whose routes do not overlap. For vehicles $k$ and $k'$, let $k \prec k'$ denotes that $k$ precedes $k'$ in the crossing order. For a vehicle $k$, let ${\prec} k$ denote the set of all vehicles preceding $k$. Vehicles already passing through or moving away from the intersection do not need to be ordered.

\begin{figure}[!t]
    \centering
    \begin{subfigure}[b]{0.7\columnwidth}
        \includegraphics[width=\columnwidth]{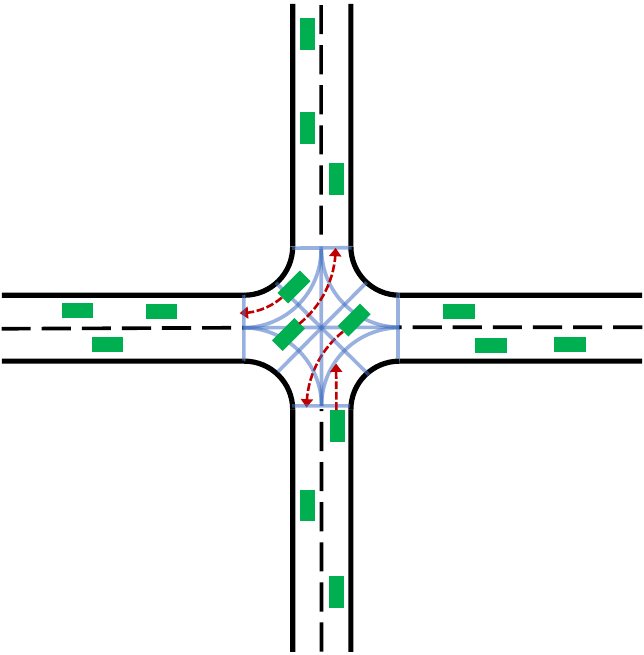}
        \caption{Intersection with reservation subzones}
    \end{subfigure}
    \begin{subfigure}[b]{0.9\columnwidth}
        \includegraphics[width=\columnwidth]{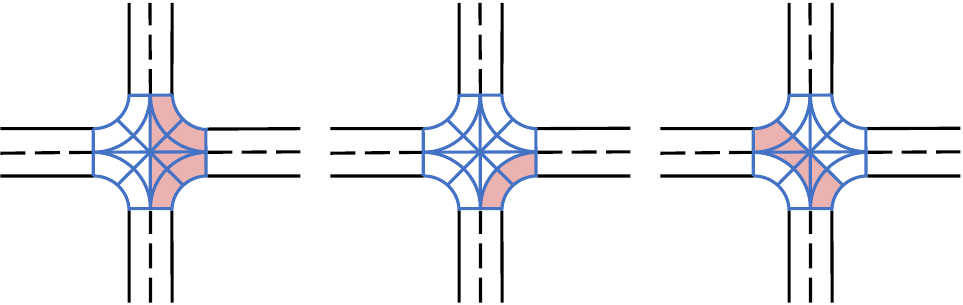}
        \caption{Subzones for straight, right-, and left-turn routes}
    \end{subfigure}
    \caption{\textbf{Geometry of our studied intersection.} Our algorithms are applicable to junctions in general, \eg merging, as the exact geometry is encoded by the start and end positions of subzones along vehicle routes, $x_{z, r}$.}
    \label{fig:intersection_subzones}
\end{figure}

\subsection{Kinematic Bicycle Model}
While we plan with the longitudinal 1D model of vehicles along their routes, all control inputs are translated to and executed on a kinematic bicycle model \cite{rajamani2011vehicle}. Here, both front and rear wheels of the vehicle are aggregated into a singular wheel at the midpoint. The control inputs are acceleration $a$ and front wheel steering angle $\delta$. The distance from the center of gravity to front and rear wheels is half of vehicle length $\ell_k$. $\psi$ denotes the heading. $\beta$ denotes the slip angle.

\begin{multline}
\dot{x}_x=v \cos (\psi+\beta) \hspace{1cm} \dot{x}_y=v \sin (\psi+\beta) \hspace{1cm} \dot{v} = a \\
 \dot{\psi}=\frac{v \cos (\beta)}{\ell_k}\tan \left(\delta\right) \hspace{1cm} \beta=\tan^{-1}\left(\frac{\tan \delta}{2}\right)\\
\label{bm_eq}
\end{multline}

\section{Optimal Crossing Order Search with MAPF}
\label{sec:method}
We extract elements of previous works on cooperative intersection crossing \cite{xu2021comparison} and design algorithms for finding the optimal crossing order from a MAPF perspective: for the high-level crossing order search, our PP and OBS algorithms integrate traffic structures. For low-level subzone reservation (akin to single-agent path planning), our KATS technique schedules arrival and departure times at reservations subzones. We use the computed crossing order to plan vehicle trajectories sequentially, with trajectory optimization. We sketch our overall method in Algorithm~\ref{alg:obs_kats}.

\subsection{Kinematic Arrival Time Scheduling (KATS)}
\label{sec:arrival_time}
MAPF algorithms rely on numerous calls to a single-agent path planner (typically A* or SIPP, which are fast but often require models with limited kinematics, like constant speed \cite{kasaura2022prioritized}). On the other hand, general mixed-integer trajectory optimization is expressive but cannot be directly used as a single-agent path planner due to the high computational overhead. Therefore, we refine the arrival times scheduling technique sketched by \cite{xu2021comparison} into a proxy for a single-agent path planner: Kinematic Arrival Time Scheduling (KATS). KATS can be efficiently invoked by high-level planners for computing an optimal crossing order.

KATS plans the subzone arrival and departure times for a vehicle $k$ on route $r$. Let $t_d({\prec} k, z)$ be the latest time that vehicles preceding $k$ occupy subzone $z$. Let $t_a(k, z) = t(x_k \geq x_{z, r, 0})$ be the arrival time of vehicle $k$ at subzone $z \in \cZ(r) \subset \cZ$ and $t_d(k, z) = t(x_k - \ell_k \geq x_{z, r, 1})$ be the departure time. KATS computes the interval $[t_a(k, z), t_d(k, z)]$ for all $z \in \cZ(r)$. The arrival time at the first subzone $z_0$ along route $r$ is computed by
\begin{equation}
    t_a(k, z_0) = \max\left \{\underline{t_a}(k, z_0), \max_{z \in \cZ(r)} \left\{t_d({\prec} k, z) - \delta t(z_0, z)\right\}\right \}
    \label{eq:t_a}
\end{equation}
where the first term $\underline{t_a}(k, z_0)$ is the minimum arrival time to $z_0$ (independent of other vehicles) and the second term is the earliest crossing start time such that the vehicle travels at \textit{constant} speed within the intersection and reaches every subzone after it becomes available. $\delta t(z_0, z) = \frac{x_{z, r, 0} - x_{z_0, r, 0}}{v_{z_0}}$ is the travel time from $z_0$ to $z$ at the \textit{maximum attainable crossing speed} $v_{z_0} \leq \overline{v}_{r, z}$. To achieve the minimum time to enter subzone $z_0$, the vehicle accelerates at $\overline a$ for as long as feasible, then travels at maximum speed $\overline{v}$ if feasible, then decelerates at $\underline a$ if needed to $\overline{v}_{r, z}$. A crossing order is \textit{infeasible} if some vehicle has insufficient distance to decelerate to $v \leq \overline{v}_{r, z}$. While KATS enforces collision-free subzones, it does not detect rear-end collisions with other vehicles along the approaching and departing lanes, and thus may be overly optimistic, as discussed in Section~\ref{sec:traj_opt}. Thus, a crossing ordering giving a following vehicle precedence over a leading vehicle may be \textit{feasible} but is unlikely to be optimal and will be pruned by heuristics below.

\begin{theorem}
If a crossing order is feasible, calling KATS in this order obtains the optimal constant-speed crossing times for all vehicles consistent with the crossing order.
\end{theorem}
\begin{proof}
(Sketch) Consider the first vehicle $k$ in the crossing order. By construction, no other acceleration strategy besides the one above allows $k$ to arrive at $z_0$ earlier than $\underline{t_a}(k, z_0)$ or with greater speed than $v_{z_0}$ above. Therefore, arriving earlier than $t_a(k, z_0)$ either contradicts the minimum arrival time or enters some $z$ before ${\prec} k$ has departed. Thus $k$ achieves the optimal delay. By induction on crossing order, KATS obtains optimal delay for all vehicles.
\end{proof}

\subsection{Prioritized Planning (PP) with Traffic Heuristics}
\label{sec:pp}
Naively optimizing the crossing order with PP \cite{erdmann1987multiple} simply samples $n_\text{orders}$ random orders, evaluates the total delay of each order with KATS, then return the best crossing order. As naive PP does not leverage traffic structures and performs poorly, we augment naive PP with the two pruning heuristics introduced by \cite{xu2021comparison} for their MCTS-based method: \textbf{1)} When sampling a random crossing order for PP, we constrain every vehicle $k$ to be sampled \textit{after} its leader (vehicle in front of $k$) in the lane. Sampling one by one from a space $\kappa$ of $\leq |\cD| = 4$ vehicles at a time, the overall search space reduces from $O(|K|!)$ to $O(4^{|K|})$ orderings. \textbf{2)} We select $k \in \kappa$ if its minimum arrival time $\underline{t_a}(k, z)$ at each subzone $z \in \cZ(r_k)$ is earlier than that of all other vehicles; if no vehicle satisfies this condition, we uniformly randomly sample a vehicle whose minimum arrival time is not later than all other vehicles at all subzones. We apply these intuitions to design our order-based search next.
\begin{algorithm}
    \caption{OBS-KATS}\label{alg:obs_kats}
    \begin{algorithmic}[0]

    \Procedure{CooperativeDriving}{}
    \For{$h = 0$ to $H$}
        \State $\cK \leftarrow$ set of all current vehicles
        \For{each newly entered vehicle $k$, with $\cK \prec k$}
            \State arrival/departure times $\leftarrow$ KATS($k$, $t_d({\prec}k, \cZ)$)
            \State \textsc{TrajOpt}($k$, ${\prec}k$, arrival/departure times)
        \EndFor

        \State execute next step along vehicle trajectories

        \If{$h \mod H_c = 0$}
            \State vehicles on entering lanes $K \subseteq \cK$
            \State crossing order $\leftarrow$ \textsc{OBS}($K$, $t_d({\prec} K, \cZ), n_\text{orders}$)
            \For {each vehicle $k$ in crossing order}
                \State arr./depart. times $\leftarrow$ \textsc{KATS}($k$, $t_d({\prec}k, \cZ)$)
                \State \textsc{TrajOpt}($k$, ${\prec}k$, arrival/departure times)
            \EndFor
        \EndIf
    \EndFor
    \EndProcedure
    \vspace{0.1cm}
    \Procedure{KATS}{$k$, $t_d({\prec} k, \cZ)$} \Comm{Section~\ref{sec:arrival_time}}
    \LeftComment{0}{\textit{$k$: vehicle $k$ to plan}}
    \LeftComment{0}{\textit{$t_d({\prec} k, \cZ)$: latest subzone departure times of ${\prec} k$}}
    \State $t_a(k, z_0) \leftarrow$ apply Equation~\ref{eq:t_a}
    \For{$z \in \cZ(r_k)$}
        \State $t_a(k, z) \leftarrow t_a(k, z_0) + \delta t(z_0, z)$
        \State $t_d(k, z) \leftarrow t_a(k, z_0) + \frac{x_{z, r_k, 1} + \ell_k - x_{z_0, r_k, 0}}{\overline{v}_{r_k, z}}$
    \EndFor
    \State \Return $\{[t_a(k, z), t_d(k, z)] \mid z \in \cZ(r_k)\}$
    \EndProcedure
    \vspace{0.1cm}
    \Procedure{OBS}{$K$, $t_d({\prec}K, \cZ)$, $n_\text{orders}$} \Comm{Section~\ref{sec:pbs}}
    \LeftComment{0}{\textit{$K$: set of vehicles to obtain a crossing order for}}
    \LeftComment{0}{\textit{$t_d({\prec} K, \cZ)$: latest subzone departure times of ${\prec} K$}}
    \LeftComment{0}{\textit{$n_\text{orders}$: number of orders to obtain}}

    \LeftComment{1}{\textit{heuristic rule 1}}
    \State $\bm{\prec} \leftarrow$ initial ordering of vehicles along lanes
    \State orders $\leftarrow$ empty list
    \Procedure{\textsc{Expand}}{$K$, $\bm{\prec}$, $n_\text{orders}$}
        \LeftComment{2}{\textit{expand a search node...}}
        \State $\kappa \leftarrow \{k \in K \mid \forall k' \in K\ k' \not \prec k\}$
        \While{$\exists k \in \kappa\ \forall k' \in \kappa_{\neq k} \left(k \ll k'\right)$}
            \State $\bm{\prec} \leftarrow \bm{\prec} \cup\ \{k \prec k'\ \forall k' \in \kappa_{\neq k}\}$
            \State let $K \leftarrow K \setminus \{k\}$
            \State update $\kappa \leftarrow \{k \in K \mid \forall k' \in K\ k' \not \prec k\}$
        \EndWhile
        \If{$\kappa = \emptyset$}
            \State construct order from $\bm{\prec}$, append to orders
            \State compute delay(order)
            \State \Return $1$
        \EndIf
        \State $n \leftarrow 0$
        \LeftComment{2}{\textit{heuristic rule 2: $k$ is closer to the intersection}}
        \State $k, k'\leftarrow$ two vehicles $\in \kappa$ s.t. $k \not\ll k'$ and $k' \not\ll k$
        \LeftComment{2}{\textit{1st child}}
        \State run \textsc{KATS} on $k'$ and any necessary $k'' \succ k'$
        \If{schedules for $k'$ and any $k''$ are feasible}
            \State $n \mathrel{\raisebox{0.19ex}{$\scriptstyle+$}}=$ \textsc{Expand}($K$, $\bm{\prec} \cup\ \{k \prec k'\}$, $\lceil\frac{n_\text{orders}}{2}\rceil$)
        \vspace{-0.4\baselineskip}
        \EndIf
        \State\textbf{if }{$n = n_\text{orders}$} \Return $n$
        \LeftComment{2}{\textit{2nd child: will be skipped if $n_\text{orders} = 1$}}
        \State run \textsc{KATS} on $k$ and any necessary $k'' \succ k$
        \If{schedules for $k$ and any $k''$ are feasible}
            \State $n \mathrel{\raisebox{0.19ex}{$\scriptstyle+$}}=$ \textsc{Expand}($K$, $\bm{\prec} \cup\ \{k' \prec k\}$, $n_\text{orders} - n$)
        \vspace{-1.2\baselineskip}
        \EndIf
        \State \Return $n$
    \vspace{0.5\baselineskip}
    \EndProcedure
    \textsc{Expand}($K$, $\bm{\prec}$, $n_\text{orders}$)
    \State \Return $\arg\min_{\text{order} \in \text{orders}}$ delay(order)
    \EndProcedure
    \end{algorithmic}
\end{algorithm}

\subsection{Order-based Search (OBS) with Traffic Heuristics}
\label{sec:pbs}
Inspired by the PBS algorithm \cite{ma2019searching} in MAPF, we design the OBS algorithm (Algorithm~\ref{alg:obs_kats}) for searching for crossing orders. While PBS searches the space of all partial orderings, we search the space of all partial orderings \textit{consistent with a total ordering of vehicles crossing each subzone}.

Each node of the OBS depth-first search tree corresponds to a set of vehicles $K$ which are yet to be ordered and an ordering ${\bm \prec}$ across all vehicles. We define $\kappa \subseteq K$ as the set of vehicles with no preceding vehicles in $K$. For two vehicles $k$ and $k'$, we define the operator $k \ll k'$ to denote the following property: the subzone departure times of $k$ and all vehicles preceding $k$ is less than the subzone arrival times of $k'$ and all vehicles succeeding $k'$ for every subzone $z \in \left(\cZ(r_k)\cup\cZ(r_{{\prec} k})\right)\cap\left(\cZ(r_{k'})\cup\cZ(r_{{\prec} {k'}})\right)$. Intuitively, if $k \ll k'$ and $\kappa = \{k, k'\}$, then $k \prec k'$ because $k$ and ${\prec} k$ crossing earlier does not delay $k'$ or ${\succ} k'$. If $k \not\ll k'$, even if $k$ departs all subzones earlier than $k'$ arrives, we cannot let $k \prec k'$ because some vehicle preceding $k$ departs some subzone later than some vehicle succeeding $k'$. If $\exists k \in \kappa\ \forall k' \in \kappa_{\neq k}$ such that $k \ll k'$, we assign precedences $k \prec \kappa_{\neq k}$, remove $k$ from $K$, and update $\kappa$ with the new $K$. Otherwise, as in PBS, we branch over the precedence of two vehicles in $\kappa$. If $\kappa$ is empty, we read the crossing order from ${\bm \prec}$.

We apply similar traffic heuristics to OBS as described for PP. To control the search duration, we limit the number of orders found to $n_\text{orders}$ total by distributing a budget of $\lceil \frac{n_\text{orders}}{2} \rceil$ orders to the first child and the remaining to the second child. This strategy allows exploration to be focused on the shallower nodes in the tree search, where decisions are more influential than decisions deeper in the tree search.

We now prove several properties about OBS.

\begin{theorem}
All orders found by OBS are crossing orders, \ie OBS is sound.
\end{theorem}
\begin{proof}
First note that precedence is only ever assigned between $k,k'\in \kappa$, whose members contain no precedence over each other by definition. Therefore, OBS never assigns an inconsistent precedence and is consistent with any initial precedence relations provided. As more precedences are assigned, some vehicle $k$ must be removed from $K$ eventually, allowing some vehicle $k' \succ k$ to join $\kappa$ eventually. By induction, \textit{every} vehicle in $K$ must eventually be added to $\kappa$, and thus be removed from $K$ eventually. Each vehicle removed from $K$ has precedence over all remaining vehicles in $K$. Therefore, the removal order from $K$ is a valid total ordering. At the leaf node, $\kappa = \emptyset$, so $K = \emptyset$ and all vehicles must be present in the total ordering returned.
\end{proof}

\begin{theorem}
Given that a crossing order exists, OBS with $n_\text{orders} = \infty$ finds the optimal \edit{constant-speed} crossing order, \ie OBS is \edit{asymptotically} optimal and complete.
\end{theorem}
\begin{proof}
We show that some branch of the OBS tree must reach an optimal crossing order, if one exists. A node in the OBS tree must add an optimal precedence relation along some branch. There are two cases:

\textbf{1)} $k \prec k'$ is added $\forall k' \in \kappa_{\neq k}$. In this case, all vehicles $k''\in K\setminus\kappa$ are already preceded by $k$ or preceded by some $k' \in \kappa_{\neq k}$. For the former, $k$ must cross earlier than $k''$ by definition. For the latter, $k \ll k' \prec k''$ implies that, at every subzone $z$, the latest subzone departure time of $k$ and all vehicles preceding $k$ is already earlier than the earliest subzone arrival time of $k'$ and $k''$. Thus, assigning $k$ to precede all other vehicles $k''\in K_{\neq k}$ does not delay the crossing of any $k''$, and giving $k$ precedence is optimal.

\textbf{2)} $k \prec k'$ is added to one child branch and $k' \prec k$ to the other. This case must be optimal because either $k \prec k'$ or $k' \prec k$ is consistent with the optimal crossing order. Without loss of generality, assume that $k \prec k'$ is consistent with optimal, then KATS must find that replanning arrivals times for $k'$ and $k''$ is feasible because ${\bm \prec}$ only has a subset of the precedence constraints of the optimal crossing order.

As every non-leaf OBS node adds at least one optimal precedence along some branch, OBS must reach the optimal leaf node because there are at most $|K|^2$ possible precedence relations total. The leaf node corresponds to a sound crossing order, as shown earlier. Thus OBS always finds the optimal crossing order and is complete.
\end{proof}

\subsection{Trajectory Optimization}
\label{sec:traj_opt}
With the total crossing ordering of vehicles, we obtain trajectories for each vehicle one-by-one, accounting for the positions of all previously planned vehicles and obeying the scheduled arrival times at subzones. KATS may be overly optimistic and inconsistent with trajectory optimization as KATS does not prevent collision between vehicles outside the intersection. Thus, following the scheduled times precisely may be infeasible. To ease infeasibility, we \textbf{1)} incrementally delay the scheduled time constraint until feasible \textbf{2)} allow all vehicles to exceed the turning speed limit except at the midpoint of a turn, which allows a vehicle to decelerate into a turn and accelerate out of a turn. Given planning horizon $T_p$, trajectory optimization for each vehicle is formulated as follows and optimized with a discretization $\text{d}t$:
\begin{multline}
    \hfill \max_{x(t), v(t)} \int_{0}^{T_p} v(t) \text{d}t\hspace{0.5cm} \text{s.t.} \hfill \\
    x(0) = 0 \hspace{1cm} v(0) = v_0 \hspace{1cm} 0 \leq v(t) \leq \overline v\\
    \underline a \leq \frac{v(t) + v(t + \text d t)}{\text d t} \leq \overline a\\
    x(t + \text d t) - x(t) = \frac{v(t) + v(t + \text d t)}{2}\\
    \underline x \leq x(t) \leq \overline x \hspace{2cm} v(t_{mid}) \leq \overline v_{r, z}\\
    \label{eq:traj_opt}
\end{multline}
where $v(t_{mid})$ is the speed at the midpoint crossing time, $\overline{x}$ is the maximum safe position of a vehicle given its subzone arrival times and leading vehicles trajectories on both the entering and exiting lane, and the minimum position constraint $\underline{x}$ ensures that the vehicle departs a subzone on schedule. To obtain the steering angles along a route, we utilize a PID controller tracking the center of the route.

\subsection{Why is Crossing Order Useful?}
We acknowledge that the optimal arrival times \textit{consistent with an optimal crossing order} does not necessarily imply optimal arrival times \textit{in general} for minimizing delay. Indeed, similar to observed by \cite{ma2019searching}, optimal arrival times may not be consistent with any crossing order. An example can be obtained by manipulating our vehicle and subzone geometries. Let the intersection be a 10 by 10 grid of square subzones, and let each vehicle be the size of one subzone. One vehicle approaches the intersection along each of the four directions, symmetrically. Clearly, the optimal arrival times is obtained by simultaneously allowing all four vehicles pass the intersection. However, these arrival times are not consistent with any crossing order, because each vehicle enters some subzone before another vehicle. With an optimal crossing order of [\textit{up}, \textit{right}, \textit{down}, \textit{left}], only the first three vehicles can enter at the same time, and \textit{left} waits for \textit{up} to finish crossing before entering their shared subzone.

Nevertheless, since using trajectory optimization as a single-agent path planner is not practical, MAPF algorithms tend to use path planners on simplified kinematics instead, as we do with KATS, resulting in a mismatch between the trajectories planned with simplified kinematics and ones planned with trajectory optimization. Obtaining a crossing order allows us to plan trajectories with complex kinematics according to the crossing order, adding delays when necessary to ease infeasibility due to the mismatch before planning subsequent vehicles. On the other hand, while a classical MAPF algorithm may find the optimal symmetric solution for the described example in simplified kinematics, a mismatch with trajectory optimization may occur resulting in infeasibility, which cannot be resolved by adding delays as doing so may conflict with other vehicles' trajectories. Therefore, crossing orders may be more robust to model mismatch between the kinematics used in MAPF and the kinematics used in trajectory optimization.

\section{Default Experimental Setup}
\label{sec:setup}
We modify HighwayEnv \cite{highway-env} to simulate the system with discretization $\text{d}t = 0.1$s for $H = 1000$ timesteps. Fig.~\ref{fig:intersection_subzones} illustrates subzone geometries. We set arrival rate $\lambda = 1500$veh/hr/lane with initial speed $v_0 = 5$m/s. Crossing order computation occurs every $H_c = 100$ steps. Each vehicle is planned for a horizon $T_p$ which is sufficient for it to reach the end of its route. Maximum speed is $\overline v = 13$m/s, with $\overline v_{r, z} = 6.5$m/s on left turns and $\overline v_{r, z} = 4.5$m/s on right turns. A vehicle goes straight, turns left, and turns right with 60\%, 20\%, and 20\% chance, respectively. Each lane has width $w_\text{lane} = 4.5$m and length $\ell_\text{lane} = 250$m. Each vehicle has length $\ell_k = 5$m and width $2$m. Each intersection is a square with edge length $5w_\text{lane}$. The left turn radius is $3w_\text{lane}$ and the right turn radius is $2w_\text{lane}$. Vehicles collide when their bounding boxes overlap; for verifying algorithmic correctness, we do not add any temporal or spatial padding around each vehicle. We run all settings on $100$ environment seeds, where we quantify the 95\% confidence interval of the mean with bootstrap sampling. All methods are implemented in Python since KATS is very fast (around 10000 calls/s), unlike single-agent path planners for classical MAPF settings, which are often implemented in C++ for efficiency. Trajectory optimization uses CVXPY \cite{diamond2016cvxpy}.

\section{Experimental Results}
\label{sec:results}
We demonstrate the effectiveness of OBS against the FIFO order \cite{dresner2008multiagent}, MCTS \cite{xu2019cooperative,xu2021comparison}, and our own PP on various intersection configurations. All methods use KATS. As no code was provided, we implement MCTS to the best of our abilities, with the same traffic heuristics as PP and OBS.

\subsection{Delay vs Crossing Order Computation Overhead}
In Fig.~\ref{fig:delay_vs_compute}, we measure the average vehicle delay as a function of the computation overhead of $n_\text{orders} \in [2^0, 2^{14}]$ for PP, $n_\text{orders} \in [2^0, 2^{13}]$ for OBS, and $n_\text{simulations} \in [2^0, 2^9]$ for MCTS. We observe that OBS is significantly stronger than PP, which is still significantly stronger than MCTS. We note that 10s per crossing order computation is a very long computation time and much longer than practical for deployment; the previous work in cooperative driving \cite{xu2021comparison} plans for around 0.1s, albeit with C++. With 10s of computation, the corresponding throughputs for the FIFO, MCTS, PP, and OBS configurations are 1740, 2050, 2080, and 2160veh/hr with confidence interval of ±20veh/hr. 

Interestingly, though the same traffic heuristics are used, the best solution quality of PP and MCTS is similar to the worst solution quality for OBS, obtained with $n_\text{orders} = 1$ and orders of magnitude less computation. We initially conjectured that the early plateauing performance of PP and MCTS may be due to the use of traffic heuristics, which may prevent finding the optimal solution. As such, we attempt to disable heuristic rule 2 for selecting the first 10 vehicles of the ordering; however, we find that the performance is significantly worse. For example, doing so for PP with computation times of 0.1s and 2s per crossing order result in average delays of 8.4s and 6.0s, respectively, both significantly worse than PP with heuristics always enabled. Thus, the difference in performance of the methods is not due to the traffic heuristics used. Rather, OBS seems to have a significant algorithmic advantage by gracefully handling partial orderings rather than searching for total orderings.

\begin{figure}[!t]
\centering
\includegraphics[width=\columnwidth]{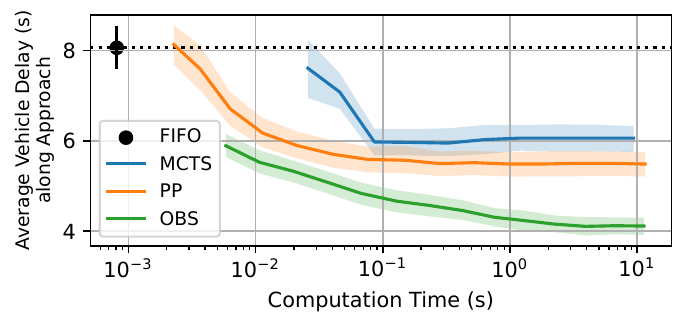}
\caption{\textbf{Delay vs computation time} per crossing order replan.}
\label{fig:delay_vs_compute}
\end{figure}

\subsection{Robustness to Intersection Configurations}
In Table~\ref{tab:configs}, we probe the robustness of our method under different conditions by varying the arrival rates $\lambda$, lane lengths $\ell_\text{lane}$, turning speeds $\overline{v}_{r, z}$ for [straight, right- and left-turn], and control horizon $H_c$. All methods are run for similar time, around 0.1s per crossing order. We find that the relative performance of all methods are consistent across configurations. The effects of arrival rate and turning speed are intuitive, so we focus on the other configurations.

For the short lane length $\ell_\text{lane} = 50$m, we observe very little gap between MCTS, PP, and OBS. This is likely due to the much smaller search space, as a $50$m lane typically contains around 3 to 4 vehicles per lane, while longer lanes contain significantly more vehicles. We see that gaps between different methods increases with the problem complexity.

Regarding the control horizon $H_c$, we observe that more frequent replans is actually slightly harmful for FIFO. At each replan, the crossing order stays constant for FIFO, but the arrival times are updated by KATS and the trajectories replanned. Since there is a mismatch between KATS and trajectory optimization, arrival times planned by KATS at step $h$ deviates from those at step $h - H_c$, shifting the constraints for trajectory optimization. This mismatch affects search-based methods as well, resulting in lowest delay at $H_c = 100$; for $H_c = 200$, delay is increased due to insufficient replans.

\begin{table}
    \caption{\textbf{Average delay (s) vs intersection configurations}}
    \label{tab:configs}
    \centering
    \begin{tabular}{ccccc}
      \toprule
        \bfseries Arrival Rate $\lambda$ (veh/hr) & \bfseries FIFO & \bfseries MCTS & \bfseries PP & \bfseries OBS\\
        \midrule
        1000 &6.6 &4.0 &3.6 & \textbf{3.2}\\
        1500 &9.6 &6.0 &5.6 & \textbf{4.7}\\
        2000 &11 &7.5 &7.0 & \textbf{6.0}\\
        $\left[500,2000,1000,1200\right]$ &6.9 &4.3 &4.0 & \textbf{3.3}\\
        \midrule
        \bfseries Lane Length $\ell_\text{lane}$ (m) & & & &\\
        \midrule
            50&3.4 &2.8 &2.7 & \textbf{2.7}\\
        100&7.5 &4.8 &4.6 &\textbf{4.1}\\
        250&9.6 &6.0 &5.6 & \textbf{4.7}\\
        500&16 &10 &9.0 & \textbf{7.8} \\
        \midrule
        \textbf{Crossing Speed} $\overline{v}_{r, z}$ \textbf{(m/s)}  & & & &\\
        \midrule
        $\left[13, 6.5, 4.5\right]$&9.6 &6.0 &5.6 & \textbf{4.7}\\
        $\left[13, 13, 13\right]$ &3.4 &2 &1.8 & \textbf{1.6}\\
        \midrule
        \bfseries Control Horizon $H_c$ (dt=0.1s) & & & &\\
        \midrule
        25&12 &8.1 &7.4 & \textbf{6.3}\\
        50&10 &6.7 &6.1 & \textbf{5.3}\\
        100&9.6 &6.0 &5.6 & \textbf{4.7}\\
        200&8.7 &6.5 &6.4 & \textbf{5.8}\\
        \bottomrule
    \end{tabular}\label{tab:configs}
\end{table}

\subsection{Delay vs Crossing Geometry}
In Fig.~\ref{fig:turn_delay}, we examine the delay for each crossing geometry (left-turn, straight, and right-turn). While OBS significantly reduces delay for routes with all geometries, it especially reduces the delay (compared to other methods) for the straight route through the intersection, which permits the highest crossing speed. For example, while the right-turn delay is higher than the straight delay for other methods, the straight delay is lower for OBS. While less apparent, similar effect can be seen for the left-turn delay. We conjecture that this OBS behavior may be due to two reasons: \textbf{1)} the delay of a vehicle going straight is the greatest if the vehicle is forced to wait at the intersection, so OBS may prioritize straight-moving vehicles \textbf{2)} a straight crossing takes the least amount of time and is less disruptive.

\begin{figure}[!t]
\centering
\includegraphics[width=\columnwidth]{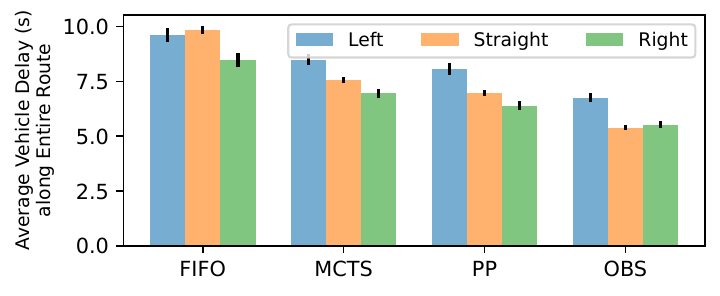}
\caption{\textbf{Delay along entire route vs crossing geometry}}
\label{fig:turn_delay}
\end{figure}

\section{Conclusions}
In this work, we seek to bridge the gap between the robotics community and the control / intelligent transportation communities. Future directions could identify other settings in traffic and robotics where the crossing order may be helpful, as well as extending the proposed algorithm to mixed traffic settings where the stochasticity of human driving behavior must be addressed. We also hope that additional insights from the robotics community may guide future algorithms for coordinating CAVs in large-scale and general traffic scenarios.

\section*{ACKNOWLEDGMENT}
This work was supported by the National Science Foundation (NSF) CAREER award (\#2239566) and the MIT Amazon Science Hub.

\bibliographystyle{IEEEtran}
\bibliography{main}

\end{document}